\documentclass[11pt,letter]{article}
\usepackage{bbding}
\fussy
\usepackage[dvipsnames]{xcolor}
\usepackage{framed, xcolor}
\usepackage{tocvsec2}
\usepackage{datetime}
\usepackage{pifont}
\usepackage{subfigure}          
\usepackage{colortbl}
\usepackage{booktabs}
\usepackage{pdfsync}
\usepackage[pagebackref=false,citecolor=newblue,colorlinks=true,linkcolor=newblue,bookmarks=false]{hyperref}

\usepackage[font=scriptsize,bf]{caption}
\usepackage{tikz,subfigure}
\usepackage[active]{srcltx}

\usepackage[margin=1in]{geometry}
\usepackage{algorithm}
\usepackage{algpseudocode}
\usepackage{datetime}
\usepackage{epsfig,amssymb,amsfonts,amsmath,amsthm}
\usepackage{multirow}
\usepackage[numbers,sort&compress,sectionbib]{natbib}

\newcommand{\Rset}{\mathbb{R}}

\newcommand{\ellip}{\mathsf{Ellip}}

\newcommand{\oracle}{\mathtt{Oracle}}
\newcommand{\calM}{\mathcal{M}}

\usepackage{amsfonts}

\usepackage{xspace}
\usepackage{tabularx}
\usepackage{pstricks}
\usepackage{setspace}

\usepackage{xcolor}

\usepackage{rotating}
\usepackage{minitoc}
\usepackage{enumerate}
\definecolor{ocre}{rgb}{0.72,0,0} 
\definecolor{newblue}{rgb}{0.2,0.2,0.6} 

\definecolor{ocre}{rgb}{0.72,0,0} 

\definecolor{babyblueeyes}{rgb}{0.63, 0.79, 0.95}

\definecolor{newgreen}{rgb}{0.53,0.66,0.42} 

\usepackage[normalem]{ulem}
\usepackage{fancybox}
\newenvironment{fminipage}%
  {\begin{Sbox}\begin{minipage}}%
  {\end{minipage}\end{Sbox}\fbox{\TheSbox}}

\usepackage{dsfont}

\newcommand{\remove}[1]{}

\newcommand{\R}{\mathbb{R}}

\newcommand{\E}{\mathbb{E}}
\newcommand{\rot}{\intercal}
\newcommand{\ce}{\mathrm{e}}
\newcommand{\tr}{\mathrm{tr}}

\renewcommand{\deg}{\mathrm{deg}}

\newcommand{\eps}{\epsilon}

\newcommand{\cp}{C_{+}}
\newcommand{\cn}{C_{-}}
\newcommand{\cs}{C}
\newcommand{\ca}{C_{|\cdot|}}

\renewcommand{\leq}{\leqslant}
\renewcommand{\geq}{\geqslant}

\newcommand{\thmref}[1]{Theorem~\ref{thm:#1}}

\newcommand{\lemref}[1]{Lemma~\ref{lem:#1}}

\newcommand{\figref}[1]{Figure~\ref{fig:#1}}

\newcommand{\secref}[1]{Section~\ref{sec:#1}}

\newcommand{\eq}[1]{\eqref{eq:#1}}

\newcommand{\Ex}[1]{\mathbb{E} \left[\,#1\,\right]}

\renewcommand{\tilde}{\widetilde}
\renewcommand{\epsilon}{\varepsilon}

\definecolor{ocre}{RGB}{150,22,11} 

\newcommand{\nnz}{\text{nnz}}

\newcommand{\mi}{I}
\global\long\def\ma{A}
\newcommand{\mb}{B}
\newcommand{\mc}{C}
\global\long\def\E{\mathbb{E}}

\global\long\def\mzero{\mathbf{0}}
\newcommand{\mm}{M}

\usepackage{setspace}

\newtheorem{thm}{Theorem}[section]  
\newtheorem{theorem}[thm]{Theorem}

\newtheorem{lem}[thm]{Lemma}
\newtheorem{lemma}[thm]{Lemma}

\newtheorem{cor}[thm]{Corollary}

\newtheorem{rem}[thm]{Remark}

\newtheorem*{rem*}{Remark}
\newtheorem{pro}[thm]{Proposition}

\newtheorem{defi}[thm]{Definition}

\newtheorem{prob}{Problem}
\renewcommand{\tilde}{\widetilde}
\usepackage{todonotes}

\numberwithin{equation}{section}

\usepackage{bm}
\newcommand{\mat}[1]{#1}

\usepackage{microtype}

\title{\textbf{An SDP-Based Algorithm for\\
 Linear-Sized Spectral Sparsification}}

\author{Yin Tat Lee\\
Microsoft Research\\ Redmond, USA\\
\texttt{yile@microsoft.com}
\and
He Sun\\
The University of Bristol\\
Bristol, UK\\
\texttt{h.sun@bristol.ac.uk}
}

\date{}

\begin{document}

\maketitle

\begin{abstract}

For any undirected and weighted graph $G=(V,E,w)$ with $n$ vertices and $m$ edges, we call a sparse subgraph $H$ of $G$, with proper reweighting of the edges,  a $(1+\varepsilon)$-spectral sparsifier if
\[
(1-\eps)x^{\rot}L_Gx\leq x^{\rot} L_{H} x\leq (1+\varepsilon) x^{\rot} L_Gx
\]
holds for any $x\in\Rset^n$,
where $L_G$ and $L_{H}$ are the respective Laplacian matrices of $G$ and $H$. Noticing that $\Omega(m)$ time is needed for any algorithm to construct a spectral sparsifier and a spectral sparsifier of $G$  requires  $\Omega(n)$ edges, a natural question is to investigate, for any constant $\eps$, if a  $(1+\eps)$-spectral sparsifier of $G$ with $O(n)$ edges can be constructed in $\tilde{O}(m)$ time, where the $\tilde{O}$ notation suppresses polylogarithmic factors. 
All previous constructions on spectral sparsification~\cite{STSS,SS,BSS,Z12,zhu15,LS15} require either super-linear number of edges or $m^{1+\Omega(1)}$ time.

  In this work we answer this question affirmatively by presenting an algorithm that, for 
  any undirected graph $G$ and $\eps>0$, outputs a $(1+\eps)$-spectral sparsifier of $G$ with $O(n/\eps^2)$ edges in $\tilde{O}(m/\eps^{O(1)})$ time. 
Our algorithm is based on three novel techniques: (1) a new potential function which is much easier to compute yet has similar guarantees as the potential functions used in previous references; 
(2) an efficient reduction from a two-sided spectral sparsifier to a one-sided spectral sparsifier; (3) constructing a one-sided spectral sparsifier by a semi-definite program.  
   
\vspace{0.5cm}


\end{abstract}

\thispagestyle{empty}

\setcounter{page}{0}

\newpage

\section{Introduction}


A sparse graph is one whose number of edges is reasonably viewed as being proportional to the number of vertices. 
Since most algorithms run faster on sparse instances of graphs and it is more space-efficient to store sparse graphs, 
it is useful to obtain a sparse representation $H$ of $G$ so that certain properties between $G$ and $H$ are preserved, see Figure~\ref{ssfig} for an illustration.
 Over the past three decades, different notions of graph sparsification have been proposed and widely used to design approximation algorithms.  For instance, a \emph{spanner} $H$
 of a graph $G$ is a subgraph of $G$ so that 
the shortest path distance between any pair of vertices is approximately preserved~\cite{Che89}.
Bencz\'{u}r and Karger~\cite{BK96} defined a \emph{cut sparsifier} of a graph $G$ to be a sparse subgraph $H$ such that the  value of any cut between $G$ and $H$ are approximately the same. 
In particular, Spielman and Teng~\cite{STSS} introduced a \emph{spectral sparsifer}, which is a sparse subgraph $H$ of an undirected graph $G$ such that many spectral properties of the Laplacian matrices between $G$ and $H$ are approximately preserved.   
Formally, for any undirected  graph $G$ with $n$ vertices and $m$ edges, we call a subgraph $H$ of $G$, with proper reweighting of the edges,  a $(1+\varepsilon)$-spectral sparsifier if 
\[
(1-\eps)x^{\rot}L_Gx\leq x^{\rot} L_{H} x\leq (1+\varepsilon) x^{\rot} L_Gx
\]
holds for any $x\in\Rset^n$,
where $L_G$ and $L_{H}$ are the respective Laplacian matrices of $G$ and $H$. 
Spectral sparsification has been proven  to be a remarkably useful tool in algorithm design, linear algebra, combinatorial optimisation, machine learning, and network analysis.

\vspace{1em}

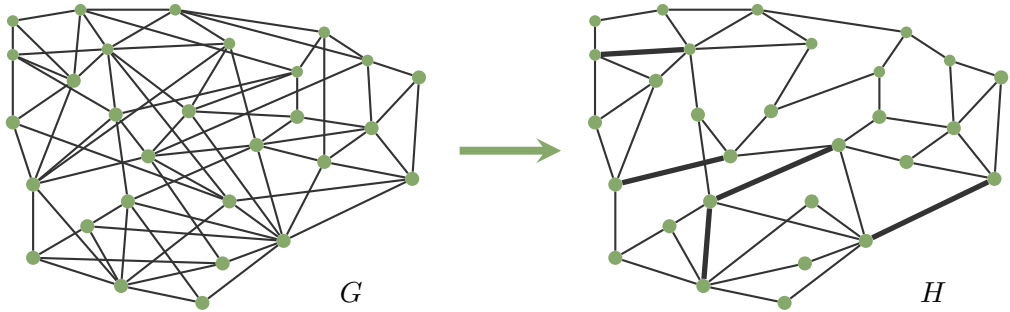
\begin{figure}[h]
\begin{center}
\begin{tikzpicture}[xscale=1.8,yscale=1.5,rounded corners=3pt,knoten/.style={fill=newgreen,color=newgreen,circle,scale=0.23},edge/.style={black!80, thick},tedge/.style={black!80, line width=2pt}]

\draw[-stealth, line width=3pt, color=newgreen] (3.5,2.35) -- (4.25,2.35);

\node[] at (2.7,1.1) {$G$};

\node[knoten] (1) at (0.2,3.5) {$1$};
\node[knoten] (2) at (0.7,3.6) {$2$};
\node[knoten] (3) at (1.4,3.6) {$3$};
\node[knoten] (4) at (2.5,3.4) {$4$};
\node[knoten] (5) at (0.2,3.2) {$5$};
\node[knoten] (6) at (0.9,3.25) {$6$};
\node[knoten] (7) at (1.8,3.3) {$7$};
\node[knoten] (8) at (2.3,3.05) {$8$};
\node[knoten] (9) at (2.82,3.15) {$9$};
\node[knoten] (10) at (3.2,3.0) {$10$};

\node[knoten] (11) at (0.2,2.6) {$11$};

\node[knoten] (12) at (0.65,2.97) {$12$};
\node[knoten] (13) at (0.96,2.67) {$13$};
\node[knoten] (14) at (1.5,2.7) {$14$};
\node[knoten] (15) at (2.3,2.65) {$15$};
\node[knoten] (16) at (2.85,2.55) {$16$};

\node[knoten] (17) at (0.35,2.05) {$17$};
\node[knoten] (18) at (1.2,2.3) {$18$};

\node[knoten] (19) at (2,2.4) {$19$};
\node[knoten] (20) at (2.5,2.25) {$20$};
\node[knoten] (21) at (3.15,2.1) {$21$};

\node[knoten] (22) at (1.05,1.9) {$22$};
\node[knoten] (23) at (1.8,1.9) {$23$};
\node[knoten] (24) at (0.75,1.68) {$24$};

\node[knoten] (25) at (2.2,1.55) {$25$};

\node[knoten] (26) at (0.35,1.4) {$26$};
\node[knoten] (27) at (1,1.15) {$27$};

\node[knoten] (28) at (1.75,1.35) {$28$};
\node[knoten] (29) at (1.6,1) {$29$};

\draw[edge] (1) -- (2);

\draw[edge] (1) -- (5);
\draw[edge] (1) -- (12);
\draw[edge] (2) -- (3);
\draw[edge] (2) -- (8);
\draw[edge] (3) -- (4);
\draw[edge] (3) -- (9);
\draw[edge] (9) -- (10);
\draw[edge] (4) -- (9);
\draw[edge] (3) -- (7);
\draw[edge] (5) -- (6);
\draw[edge] (2) -- (12);
\draw[edge] (2) -- (6);
\draw[edge] (6) -- (7);
\draw[edge] (4) -- (8);
\draw[edge] (5) -- (11);
\draw[edge] (5) -- (12);
\draw[edge] (5) -- (13);
\draw[edge] (11) -- (12);
\draw[edge] (11) -- (17);
\draw[edge] (17) -- (26);
\draw[edge] (13) -- (17);
\draw[edge] (6) -- (12);
\draw[edge] (6) -- (13);
\draw[edge] (6) -- (14);
\draw[edge] (7) -- (14);
\draw[edge] (7) -- (19);
\draw[edge] (14) -- (15);
\draw[edge] (15) -- (16);
\draw[edge] (9) -- (16);
\draw[edge] (10) -- (16);
\draw[edge] (10) -- (21);
\draw[edge] (16) -- (21);
\draw[edge] (20) -- (21);
\draw[edge] (21) -- (25);
\draw[edge] (25) -- (29);
\draw[edge] (25) -- (28);
\draw[edge] (27) -- (29);
\draw[edge] (26) -- (27);
\draw[edge] (27) -- (28);
\draw[edge] (26) -- (28);
\draw[edge] (23) -- (25);
\draw[edge] (21) -- (23);
\draw[edge] (20) -- (25);
\draw[edge] (19) -- (20);

\draw[edge] (24) -- (26);
\draw[edge] (22) -- (24);
\draw[edge] (24) -- (25);
\draw[edge] (11) -- (23);

\draw[edge] (18) -- (23);
\draw[edge] (17) -- (18);
\draw[edge] (18) -- (19);
\draw[edge] (16) -- (19);
\draw[edge] (15) -- (19);
\draw[edge] (19) -- (25);
\draw[edge] (19) -- (22);

\draw[edge] (22) -- (27);
\draw[edge] (13) -- (22);
\draw[edge] (8) -- (13);
\draw[edge] (8) -- (15);
\draw[edge] (9) -- (18);
\draw[edge] (14) -- (18);
\draw[edge] (13) -- (18);
\draw[edge] (7) -- (17);

\draw[edge] (4) -- (20);
\draw[edge] (8) -- (14);
\draw[edge] (17) -- (27);
\draw[edge] (12) -- (17);

\draw[edge] (24) -- (27);
\draw[edge] (18) -- (28);
\draw[edge] (14) -- (25);
\draw[edge] (6) -- (23);
\draw[edge] (22) -- (29);
\draw[edge] (23) -- (27);
\draw[edge] (3) -- (6);
\draw[edge] (16) -- (20);
\draw[edge] (22) -- (25);

\begin{scope}[xshift=4.3cm]

\node[] at (2.7,1.1) {$H$};

\node[knoten] (1) at (0.2,3.5) {$1$};
\node[knoten] (2) at (0.7,3.6) {$2$};
\node[knoten] (3) at (1.4,3.6) {$3$};
\node[knoten] (4) at (2.5,3.4) {$4$};
\node[knoten] (5) at (0.2,3.2) {$5$};
\node[knoten] (6) at (0.9,3.25) {$6$};
\node[knoten] (7) at (1.8,3.3) {$7$};
\node[knoten] (8) at (2.3,3.05) {$8$};
\node[knoten] (9) at (2.82,3.15) {$9$};
\node[knoten] (10) at (3.2,3.0) {$10$};

\node[knoten] (11) at (0.2,2.6) {$11$};

\node[knoten] (12) at (0.65,2.97) {$12$};
\node[knoten] (13) at (0.96,2.67) {$13$};
\node[knoten] (14) at (1.5,2.7) {$14$};
\node[knoten] (15) at (2.3,2.65) {$15$};
\node[knoten] (16) at (2.85,2.55) {$16$};

\node[knoten] (17) at (0.35,2.05) {$17$};
\node[knoten] (18) at (1.2,2.3) {$18$};

\node[knoten] (19) at (2,2.4) {$19$};
\node[knoten] (20) at (2.5,2.25) {$20$};
\node[knoten] (21) at (3.15,2.1) {$21$};

\node[knoten] (22) at (1.05,1.9) {$22$};
\node[knoten] (23) at (1.8,1.9) {$23$};
\node[knoten] (24) at (0.75,1.68) {$24$};

\node[knoten] (25) at (2.2,1.55) {$25$};

\node[knoten] (26) at (0.35,1.4) {$26$};
\node[knoten] (27) at (1,1.15) {$27$};

\node[knoten] (28) at (1.75,1.35) {$28$};
\node[knoten] (29) at (1.6,1) {$29$};

\draw[edge] (1) -- (2);
\draw[edge] (2) -- (3);
\draw[edge] (2) -- (6);
\draw[edge] (3) -- (4);
\draw[edge] (4) -- (9);
\draw[edge] (9) -- (10);
\draw[edge] (9) -- (16);
\draw[edge] (10) -- (16);
\draw[edge] (10) -- (21);
\draw[edge] (16) -- (21);
\draw[edge] (20) -- (21);
\draw[tedge] (21) -- (25);
\draw[edge] (25) -- (29);
\draw[edge] (27) -- (29);
\draw[edge] (26) -- (27);
\draw[edge] (27) -- (28);
\draw[edge] (25) -- (28);
\draw[edge] (23) -- (27);
\draw[edge] (19) -- (25);
\draw[edge] (19) -- (20);

\draw[edge] (23) -- (25);
\draw[tedge] (19) -- (22);
\draw[tedge] (22) -- (27);
\draw[edge] (22) -- (24);
\draw[edge] (24) -- (26);
\draw[tedge] (17) -- (18);
\draw[edge] (13) -- (22);

\draw[edge] (1) -- (5);
\draw[tedge] (5) -- (6);
\draw[edge] (5) -- (11);
\draw[edge] (5) -- (12);
\draw[edge] (6) -- (12);
\draw[edge] (11) -- (17);
\draw[edge] (11) -- (12);
\draw[edge] (12) -- (17);
\draw[edge] (17) -- (26);
\draw[edge] (6) -- (13);
\draw[edge] (6) -- (3);
\draw[edge] (3) -- (7);
\draw[edge] (6) -- (7);
\draw[edge] (7) -- (14);
\draw[edge] (8) -- (14);
\draw[edge] (8) -- (15);
\draw[edge] (15) -- (16);
\draw[edge] (15) -- (19);
\draw[edge] (18) -- (19);
\draw[edge] (13) -- (18);
\draw[edge] (14) -- (18);
\draw[edge] (16) -- (20);
\draw[edge] (4) -- (8);
\draw[edge] (24) -- (27);
\draw[edge] (22) -- (25);
\end{scope}

\end{tikzpicture}
\end{center}

\caption{The graph sparsification is a reweighted subgraph $H$ of an original graph $G$ such that certain properties are preserved. These subgraphs are  sparse, and are more space-efficient to be stored than the original graphs. The picture above uses the thickness of edges in $H$ to represent their weights.
\label{ssfig}}
\end{figure}

In the seminal work on spectral sparsification, Spielman and Teng~\cite{STSS}\ showed that, for any undirected graph $G$ of $n$ vertices, a spectral sparsifier of $G$ with only  $O(n\log^{c}n/\varepsilon^2)$ edges 
exists and can be constructed in nearly-linear time\footnote{We say a graph algorithm runs in nearly-linear time if the algorithm runs in $O(m\cdot \mathrm{poly}\log n)$ time, where $m$ and $n$ are the number of edges and vertices of the input graph.}, where $c\geq 2$ is some constant. Both the runtime of their algorithm  and the number of edges in the output graph involve large poly-logarithmic factors, and this motivates a sequence of  simpler and faster constructions of spectral sparsifiers with fewer edges~\cite{BSS,zhu15,LS15}.  In particular, since any constant-degree expander graph of $O(n)$ edges is a spectral sparsifier of an $n$-vertex complete graph, a natural question is to study, for any $n$-vertex undirected  graph $G$ and constant $\eps>0$, if a $(1+\eps)$-spectral sparsifier of $G$ with $O(n)$ edges can be constructed in nearly-linear time. Being considered as one of the most important open question about spectral sparsification by Batson et al.~\cite{SSSurvey}, there has been many efforts  for fast constructions of linear-sized spectral sparsifiers, e.g. \cite{zhu15,LS15}, however the original problem posed in \cite{SSSurvey} has remained open.

In this work we answer this question affirmatively by 
presenting the first nearly-linear time algorithm for constructing a linear-sized spectral sparsifier. The formal description of our result is as follows:

\begin{thm} \label{thm:maingraph}Let $G$ be any undirected graph with $n$ vertices and $m$ edges. For any $0<\varepsilon<1$, there is an algorithm that runs in $\tilde{O}\left(m/\eps^{O(1)}\right)$ work, $\tilde{O}\left(1/\eps^{O(1)}\right)$ depth, and produces a $(1+\eps)$-spectral sparsifier of $G$ with $O\left(n/\eps^2\right)$ edges\footnote{Here, the notation $\tilde{O}(\cdot)$ hides a factor of $\log^cn$ for some positive constant $c$.}. 
\end{thm}

\thmref{maingraph}  shows that  a linear-sized spectral sparsifier can be constructed in nearly-linear time in a single machine setting, and in polylogarithmic time in a parallel setting. The same algorithm can be applied to the matrix setting, whose result is summarised as follows:

\begin{thm} \label{thm:mainmatrix}
Given a set of $m$  \textsf{PSD} matrices $\{M_i\}_{i=1}^m$, where $M_i\in\Rset^{n\times n}$. Let $M = \sum_{i=1}^m M_i$ and $Z=\sum_{i=1}^m\mathrm{nnz}(M_i)$, where $\mathrm{nnz}(M_i)$ is the number of non-zero entries in $M_i$.
For any $1>\varepsilon>0$, there is an algorithm that runs in $\tilde{O}\left((Z + n^\omega)/\eps^{O(1)}\right)$ work, $\tilde{O}\left(1/\eps^{O(1)}\right)$ depth and produces a $(1+\eps)$-spectral sparsifier of $M$ with $O\left(n/\eps^2\right)$ components, i.e., there is an non-negative coefficients $\{c_i\}_{i=1}^m$ such that
$
\left| \{ c_i | c_i\neq 0\}  \right| =O\left(n/\eps^2\right)
$, and
\begin{equation}
(1-\eps)\cdot M\preceq 
\sum_{i=1}^m c_i M_i \preceq
(1+\eps)\cdot M.
\end{equation}
Here  $\omega$ is the matrix multiplication constant.
\end{thm}
\subsection{Related work}


In the seminal paper on spectral sparsification, Spielman and Teng~\cite{STSS}  showed that a spectral sparsifier of any undirected graph $G$ can be constructed by decomposing $G$ into multiple nearly expander graphs, and sparsifying each subgraph individually. This method leads to the first nearly-linear time algorithm for constructing a spectral sparsifier with $O(n\log^{c}n/\varepsilon^2)$ edges for some $c\geq 2$. 
However, both the runtime of their algorithm and the number of edges in the output graph involve large poly-logarithmic factors.  Spielman and Srivastava~\cite{SS} showed that a $(1+\eps)$-spectral sparsifier of $G$  with $O(n\log n/\eps^2)$ edges can be constructed by sampling the edges of $G$ with probability proportional to their effective resistances, which is conceptually much simpler than the algorithm presented in \cite{STSS}.

 Noticing that any constant-degree expander graph of $O(n) $ edges is a spectral sparsifier of an $n$-vertex complete graph, Spielman and Srivastava~\cite{SS} asked 
if any $n$-vertex graph has a spectral sparsifier with $O(n)$ edges.  To answer this question, Batson, Spielman and Srivastava~\cite{BSS} presented a polynomial-time algorithm that, for any undirected graph $G$ of $n$ vertices, produces a spectral sparsifier of $G$ with $O(n)$ edges.  At a high level, their algorithm, a.k.a. the \textsf{BSS algorithm}, proceeds for $O(n)$ iterations, and in each iteration \emph{one edge} is chosen deterministically to ``optimise'' the change of some potential function. Allen-Zhu et al.~\cite{zhu15} noticed that a less ``optimal" edge, based on a different potential function, can be found in almost-linear time and this leads to an almost-quadratic time algorithm. Generalising their techniques, Lee and Sun~\cite{LS15} showed that a linear-sized spectral sparsifier can be constructed in time $O\left(m^{1+c}\right)$ for an arbitrary small constant $c$.
All  of these algorithms proceed for $\Omega(n^c)$ iterations, and
every iteration takes  $\Omega(m^{1+c})$ time for some constant $c>0$.  Hence, to break the $\Omega(m^{1+c})$ runtime barrier faced in all previous constructions multiple new techniques are needed.


\subsection{Organisation}

The remaining part of the paper is organised as follows. We  introduce necessary notions about matrices and graphs in  \secref{pre}. In \secref{overview} we  overview our algorithm and proof techniques.  For readability, more detailed discussions and technical proofs 
 are deferred to \secref{details}.
\section{Preliminaries\label{sec:pre}}

\subsection{Matrices}

For any $n\times n$ real and symmetric matrix $A$, let $\lambda_{\min}(A)=\lambda_1(A)\leq \cdots\leq \lambda_n(A)=\lambda_{\max}(A)$ be the eigenvalues of $A$,  where $\lambda_{\min}(A)$ and $\lambda_{\max}(A)$ represent the minimum and maximum eigenvalues of $A$. 
We call a matrix $A$  positive semi-definite~(\textsf{PSD}) if  $x^{\rot}Ax\geq 0$ holds for any $x\in\mathbb{R}^n$, and a matrix $A$ positive definite  if  $x^{\rot}Ax> 0$ holds for any $x\in\mathbb{R}^n$.  For any positive definite matrix $A$, we define the corresponding ellipsoid by
$
\mathsf{Ellip}(A)\triangleq\left\{ x :~ x^{\rot}A^{-1}x\leq 1 \right\}.
$ 

\subsection{Graph Laplacian}

Let
  $G=(V,E,w)$ be a connected,   undirected and weighted  graph with $n$ vertices, $m$ edges, and weight function $w: E\rightarrow \mathbb{R}_{\geq 0}$. We fix an arbitrary orientation of the edges in $G$, and let $B\in\Rset^{m\times n}$ be the \emph{signed edge-vertex incidence matrix} defined by 
 \begin{equation} \nonumber
B_G(e,v)=\left\{ \begin{aligned}
         1 & \qquad \mbox{if $v$ is $e$'s head,} \\
         -1 & \qquad \mbox{if $v$ is $e$'s tail,} \\
                 0&\qquad \mbox{otherwise.}
                          \end{aligned} \right.
                          \end{equation}
We  define an $m\times m$ diagonal matrix $W_G$ by $W_G(e,e)=w_e$ for any edge $e\in E[G]$.

 The Laplacian matrix of $G$ is an $n\times n$ matrix $L$ defined by
\begin{equation} \nonumber
L_G(u,v)=\left\{ \begin{aligned}
         -w(u,v) & \qquad \mbox{if $u\sim v$,} \\
         \deg(u) & \qquad \mbox{if $u=v$,} \\
                 0&\qquad \mbox{otherwise,}
                          \end{aligned} \right.
                          \end{equation}
where $\deg(v)=\sum_{u\sim v} w(u,v)$.
It is easy to verify 
 that \[
x^{\rot}L_G x = x^{\rot}B_G^{\rot}W_GB_Gx =\sum_{u\sim v} w_{u,v}(x_u-x_v)^2 \geq 0
\] holds
for any $x\in \mathbb{R}^n$. Hence, the Laplacian matrix of any undirected graph is a \textsf{PSD} matrix.  
Notice that, by setting $x_u=1$ if $u\in S$ and $x_u=0$ otherwise, $x^{\rot}L_Gx$ equals to the value of the cut between $S$ and $V\setminus S$. Hence, a spectral sparsifier is a stronger notion than a cut sparsifer.


%

\subsection{Other notations}
For any sequence $\{\alpha_i\}_{i=1}^m$, we use $\mathrm{nnz}(\alpha)$ to denote the number of non-zeros in $\{\alpha_i\}_{i=1}^m$.
For any two matrices $A$ and $B$, we write $A\preceq B$ to represent $B-A$ is \textsf{PSD}, and $A\prec B$ to represent $B-A$
 is positive definite.  For any two matrices $A$ and $B$ of the same dimension, let $A\bullet B\triangleq \tr\left(A^{\rot}B\right)$, and
 \[
 A\oplus B =
 \left[ \begin{array}{cc}
A & \mathbf{0}  \\
\mathbf{0} & B \\
 \end{array} \right].
\]

\section{Overview of our algorithm\label{sec:overview}}

Without loss of generality we study the problem of sparsifying the sum of \textsf{PSD} matrices. The one-to-one correspondence between the
construction of a graph sparsifier and the following Problem~\ref{prob:general} was presented in \cite{BSS}.

\begin{prob}\label{prob:general}
Given a set $S$ of $m$ \textsf{PSD} matrices 
 $M_1,\cdots,M_m$ with $\sum_{i=1}^m M_i=I$ and $0<\eps<1$, find non-negative coefficients $\{c_i\}_{i=1}^m$ such that
$
\left| \{ c_i | c_i\neq 0\}  \right| =O\left(n/\eps^2\right)
$, and
\begin{equation}\label{eq:sscondition}
(1-\eps)\cdot I\preceq 
\sum_{i=1}^m c_i M_i \preceq
(1+\eps)\cdot I.
\end{equation}
\end{prob}

For intuition, one can think all $M_i$ are rank-1 matrices, i.e., $M_i = v_{i} v^{\rot}_{i}$ for some $v_i\in\Rset^n$.
Given the correspondence between \textsf{PSD} matrices and ellipsoids, Problem~\ref{prob:general} essentially asks to use $O(n/\eps^2)$ vectors from $S$ to construct an ellipsoid, whose \emph{shape} is close to be  a sphere.
To construct such an ellipsoid with desired shape, all previous algorithms~\cite{BSS,zhu15, LS15} proceed by iterations:  in each iteration $j$ the
algorithm chooses  one or more vectors, denoted by $v_{j_1},\cdots, v_{j_k}$, and adds $\Delta_j\triangleq\sum_{t=1}^{k} v_{j_t} v^{\rot}_{j_t}$ to the 
 currently constructed matrix by setting $A_j=A_{j-1}+\Delta_j$.  To control the shape of the constructed ellipsoid,  two barrier values, the \emph{upper barrier} $u_j$ and the \emph{lower barrier} $\ell_j$, are maintained  such that  the constructed ellipsoid $\ellip(A_j)$ 
is  sandwiched between  the \emph{outer} sphere $u_j\cdot I$ and the  \emph{inner} sphere $\ell_j\cdot I$ for any iteration $j$. That is, the following invariant always maintains:
\begin{equation}\label{eq:invariant}
\ell_j\cdot I\prec A_j\prec  u_j\cdot I.
\end{equation}
To ensure  \eq{invariant} holds, two barrier values $\ell_j$ and $u_j$ are increased properly after each iteration, i.e., 
\[
u_{j+1}= u_j +\delta_{u,j}, \qquad \ell_{j+1} = \ell_j + \delta_{\ell,j}
 \]
 for some positive values $\delta_{u,j}$ and $\delta_{\ell,j}$. The algorithm  continues this process, until  after $T$ iterations $\ellip(A_T)$ is close to be a sphere. This implies that $A_T$ is a solution of  
 Problem~\ref{prob:general}, see \figref{bsspic} for an illustration.
  \begin{figure}[h]
\begin{center}
\begin{tikzpicture}[xscale=1.5,yscale=1.5,rounded corners=3pt,
rknoten/.style={fill=ocre!70,color=ocre!70,draw=ocre!70,circle,scale=0.35},edge/.style={black, thick},
gknoten/.style={fill=newgreen,color=newgreen,draw=newgreen,circle,scale=0.45},edge/.style={black, thick},
bknoten/.style={fill=newblue,color=newblue,draw=newblue,circle,scale=0.35},edge/.style={black, thick},
redge/.style={draw=ocre!70,thick},
gedge/.style={draw=newgreen,thick},
bedge/.style={draw=newblue,thick},
edge/.style={black, thick}]

\draw[-stealth, color=newgreen, line width=3pt] (1.8,0.6) -- (2.6,0.6);

\node () at (1,-0.6) {
\footnotesize{\textbf{Iteration $j$}
}};

\node () at (3.7,-0.6) {
\footnotesize{\textbf{Iteration $j+1$
}}};

\node () at (8.1,-0.6) {
\footnotesize{\textbf{Final iteration $T$
}}};

\draw[-stealth, color=newgreen, line width=3pt] (4.6,0.6) -- (5.4,0.6);

\node[gknoten] (1) at (5.6,0.6){};

\node[gknoten] (1) at (5.8,0.6){};

\node[gknoten] (1) at (6,0.6){};

\draw[-stealth, color=newgreen, line width=3pt] (6.2,0.6) -- (7,0.6);

\begin{scope}[xshift=-0.9cm, yshift=-0.3cm]

\begin{scope}[xshift=0.3cm]

\draw[draw=newgreen!90, fill=newgreen!30, very thick] (1.5,1) ellipse (.65cm and .65cm);

\draw[rotate around={15:(1.5,1)}, draw=gray, fill=gray!80,  very thick] (1.5,1) ellipse (.35cm and .6cm);
\end{scope}

\begin{scope}[xshift=0.3cm]
\draw[draw=gray!90, fill=gray!30, very thick] (1.5,1) ellipse (.3cm and .3cm);
\end{scope}

\end{scope}

\begin{scope}[xshift=1.8cm, yshift=-0.3cm]

\begin{scope}[xshift=0.3cm]

\draw[draw=newgreen!90, fill=newgreen!30, very thick](1.5,1) ellipse (.75cm and .75cm);

\draw[rotate around={5:(1.5,1)}, draw=gray, fill=gray!80, very thick] (1.5,1) ellipse (.55cm and .69cm);

\end{scope}

\begin{scope}[xshift=0.3cm]
\draw[draw=gray!90, fill=gray!30, very thick] (1.5,1) ellipse (.45cm and .45cm);
\end{scope}

\end{scope}

\begin{scope}[xshift=6.3cm, yshift=-0.3cm]

\begin{scope}[xshift=0.3cm]

\draw[draw=newgreen!90, fill=newgreen!30, very thick] (1.5,1) ellipse (1cm and 1cm);

\draw[rotate around={15:(1.5,1)}, draw=gray, fill=gray!80, opacity=0.9,very thick] (1.5,1) ellipse (.85cm and .95cm);

\end{scope}

\begin{scope}[xshift=0.3cm]
\draw[draw=gray!90, fill=gray!30, very thick](1.5,1) ellipse (.8cm and .8cm);
\end{scope}

\end{scope}

\end{tikzpicture}
\end{center}

\caption{Illustration of the algorithms for constructing a linear-sized spectral sparsifier. Here, the light grey and green circles in  iteration $j$ represent the spheres $\ell_j\cdot I$ and $u_j\cdot I$, and the dark grey ellipse sandwiched between the two circles corresponds to the constructed ellipsoid in iteration $j$. After each iteration $j$,  the algorithm increases the value of $\ell_j$ and $u_j$
by some $\delta_{\ell,j}$ and $\delta_{u,j}$
 so that the invariant \eq{invariant}  holds in iteration $j+1$. 
This process is repeated for $T$ iterations, so that the final constructed ellipsoid  is close to be a sphere.
\label{fig:bsspic}}
\end{figure}
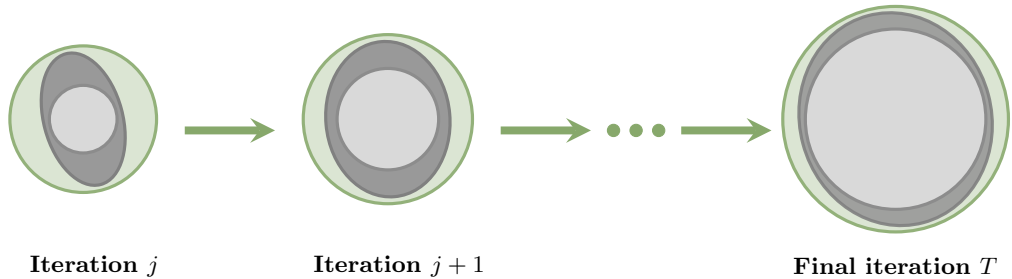

However, turning the scheme  described above into an efficient algorithm 
we need to consider the following  issues:
\begin{itemize}
\item Which vectors  should we pick in each iteration?
\item How many vectors can be added in each iteration? 
\item How should we update $u_j$ and $\ell_j$ properly so that the invariant \eq{invariant} always holds? 
\end{itemize}
These three questions closely relate to each other: on  one hand, one can always  pick a single ``optimal'' vector in each iteration based on some metric, and such conservative approach requires a linear number of iterations $T=\Omega(n/\eps^2)$ and
super-quadric time for each iteration. On the other hand, one can choose multiple less ``optimal" vectors to construct $\Delta_j$ in iteration $j$, but this makes the update of barrier values more challenging to ensure the invariant \eq{invariant} holds. 
Indeed, the previous constructions~\cite{zhu15, LS15} speed up their algorithms at the cost of increasing the sparsity, i.e., the number of edges in a sparsifier, by more than a multiplicative constant.

%

To address these, we introduce three novel techniques for constructing a spectral sparsifier:
First of all, we define a new potential function which is much easier to compute yet has similar guarantee as the potential function introduced in~\cite{BSS}.
Secondly we show that solving Problem~\ref{prob:general} with two-sided constraints in \eq{sscondition} can be reduced to a similar problem  with only one-sided constraints.  Thirdly we prove that the problem with one-sided constraints can be solved by a semi-definite program.

\subsection{A new potential function\label{sec:potential}}

To ensure that the constructed ellipsoid $A$ is always inside the outer sphere $u\cdot I$, we introduce a potential function $\Phi_u(A)$ defined by 
\[
\Phi_u(A)\triangleq \tr\exp\left((uI-A)^{-1}\right)=\sum_{i=1}^n \exp\left( \frac{1}{u-\lambda_i(A)} \right).
\]
It is easy to see that, when  $\ellip(A)$ gets closer to the outer sphere, $\lambda_i(u\cdot I-A)$ becomes smaller and the value of  $\Phi_u(A)$ increases. Hence, a bounded value of $\Phi_u(A)$ ensures that $\ellip(A)$ is inside the sphere $u\cdot I$. For the same reason, we introduce a potential function $\Phi_{\ell}(A)$
defined by 
\[
\Phi_{\ell}(A)\triangleq\tr\exp\left((A-\ell I)^{-1}\right)=\sum_{i=1}^n \exp\left( \frac{1}{\lambda_i(A)-\ell} \right).
\]
to ensure that the inner sphere is always inside $\ellip(A)$. We also define 
\begin{equation}\label{eq:defpotential}
\Phi_{u,\ell}(A)\triangleq\Phi_{u}(A) + \Phi_{\ell}(A),
\end{equation}
as a bounded value of $\Phi_{u,\ell}(A)$
 implies that the two events occur simultaneously. Our goal is to design a proper update rule 
 to construct $\{A_j\}$ inductively, so that $\Phi_{u_j, \ell_j}(A_j)$ is monotone non-increasing after each iteration. Assuming this, a bounded value of the initial potential function guarantees that the invariant \eq{invariant} always holds.
 
 To analyse the change of the potential function, we first notice that
 \[
\Phi_{u,\ell}(\ma+\Delta)  \geq   \Phi_{u,\ell}(\ma) + \tr\left( \mathrm{e}^{(u\mi-\ma)^{-1}} (u\mi-\ma)^{-2}\Delta\right) - \tr\left(\mathrm{e}^{(\ma-\ell\mi)^{-1}}(\ma-\ell\mi)^{-2}\Delta\right)
\]
by the convexity of the function $\Phi_{u,\ell}$.
We  prove that, as long as the matrix $\Delta$ satisfies $0\preceq\Delta\preceq\delta(uI-\ma)^{2}$ and $0\preceq\Delta\preceq\delta(\ma-\ell I)^{2}$ for some small $\delta$, the first-order approximation gives  a good approximation.

\begin{lem}\label{lem:potentialchange}
Let $A$ be a symmetric matrix. Let $u,\ell$ be the barrier values such that $u-\ell\leq1$ and $\ell\mi\prec\ma\prec u\mi$.
Assume
that $\Delta\succ 0$,
 $\Delta\preceq\delta(uI-\ma)^{2}$ and $\Delta\preceq\delta(\ma-\ell I)^{2}$
for $\delta\leq 1/10$. Then, it holds that
\begin{align*}
\Phi_{u,\ell}(\ma+\Delta)  \leq  & \Phi_{u,\ell}(\ma) + (1+2\delta)\tr\left( \mathrm{e}^{(u\mi-\ma)^{-1}} (u\mi-\ma)^{-2}\Delta\right)\\
 &  \qquad -(1-2\delta)\tr\left(\mathrm{e}^{(\ma-\ell\mi)^{-1}}(\ma-\ell\mi)^{-2}\Delta\right).
\end{align*}
\end{lem}

We remark that this is not the first paper to use a potential function to guide the growth of the ellipsoid. In~\cite{BSS}, the potential function
\begin{equation}\label{eq:oldpotential}
\Lambda_{u,\ell,p}(\ma) = \tr\left((uI-A)^{-p}\right) + \tr\left((A-\ell I)^{-p}\right)
\end{equation}
is used
with $p= 1$. The main drawback is that $\Lambda_{u,\ell,1}$ does not differentiate the
following two cases: 
\begin{itemize}
\item Multiple eigenvalues of $A$ are close to the boundary (both $u$
and $\ell$).
\item One of the eigenvalues of $A$ is very close to the boundary (either
$u$ or $\ell$).
\end{itemize}
It is known that, when one of the eigenvalues of $A$ is very close to the boundary, it is more difficult to find 
an ``optimal'' vector. 
\cite{zhu15} shows this problem can be alleviated by using $p \gg 1$. 
However, this choice of $p$ makes the function $\Lambda_{u,\ell,p}$ less smooth and hence one has to take a smaller step size $\delta$, as shown in the following lemma by \cite{zhu15}.

\begin{lem}[\cite{zhu15}]
Let $A$ be a symmetric matrix and $\Delta$ be a rank-1 matrix. Let $u,\ell$ be the barrier values such that $\ell\mi\prec\ma\prec u\mi$.
Assume that $\Delta\succ 0$, $\Delta\preceq\delta(uI-\ma)$ and $\Delta\preceq\delta(\ma-\ell I)$
for $\delta\leq 1/(10 p)$ and $p \geq 10$. Then, it holds that 
\begin{align*}
\Lambda_{u,\ell,p}(\ma+\Delta) \leq  & \Lambda_{u,\ell,p}(\ma) + p (1+p \delta)\tr\left((u\mi-\ma)^{-(p+1)}\Delta\right)- p(1-p \delta)\tr\left((\ma-\ell\mi)^{-(p+1)}\Delta\right).
\end{align*}
\end{lem}

Notice that, comparing with the potential function \eq{oldpotential}, our new potential function \eq{defpotential} blows up much faster when the eigenvalues of $A$ are closer to the boundaries $\ell$ and $u$.
This allows the problem of finding an ``optimal'' vector much easier than using $\Lambda_{u,\ell,p}$. At the same time, we avoid the problem of taking a small step $\Delta\preceq 1/p\cdot (uI-\ma)$ by taking a ``non-linear" step $\Delta\preceq (uI-\ma)^2$. As there cannot be too many eigenvalues close to the boundaries, this ``non-linear'' step allows us to take a large step except on a few directions.

\subsection{A simple construction based on $\oracle$\label{sec:reduction}}

The second technique we introduce is the reduction from 
a  spectral sparsifier with two-sided constraints
 to the one with one-sided constraints.  Geometrically,
it is equivalent to require the constructed ellipsoid inside another ellipsoid, instead of being sandwiched between two spheres as depicted in \figref{bsspic}.
Ideally, we want to reduce the two-sided problem to the following problem: for a set of \textsf{PSD} matrices $\calM=\{M_i \}_{i=1}^m$ such that $\sum_{i=1}^m{M_i} \preceq I$,  find a sparse representation $\Delta=\sum_{i=1}^m \alpha_i M_i$ with small $\nnz(\alpha)$ such that $\Delta\preceq I$. 
However, in the reduction we use such matrix  $\Delta=\sum_{i=1}^m \alpha_i M_i$ to update $A_j$ and we need the length of $\ellip(\Delta)$ is large on the direction that $\ellip(A_j)$ is small.
 To encode this information, we define the generalised one-sided problem as follows:




\begin{defi}[\textbf{One-sided Oracle}]\label{deforacle}
Let  $\mzero \preceq B \preceq I$, $\cp \succeq \mzero, \cn \succeq \mzero$ be  symmetric matrices,  
$\mathcal{M}=\{M_i\}_{i=1}^m$ be  a set of matrices such that $\sum_{i=1}^{m}\mm_{i} = \mi$. 
We call a randomised algorithm $\oracle\left(\calM, B,\cp, \cn \right)$ a one-sided oracle with speed $S\in (0,1]$ and error $\eps>0$, if $\oracle\left(\calM, B,\cp, \cn \right)$  outputs  a matrix $\Delta=\sum_{i=1}^{m}\alpha_{i}\mm_{i}$ such that 
\begin{enumerate}
\item $\mathrm{nnz}(\alpha)\leq \lambda_{\min} (\mb) \cdot \tr\left(\mb^{-1}\right)$.
\item $\Delta\preceq\mb$ and $\alpha_i \geq 0$ for all $i$.
\item $\E\left[\cs\bullet\Delta\right] \geq S \cdot \lambda_{\min}(\mb) \cdot \tr(\cs) - \eps S \cdot \lambda_{\min}(\mb) \cdot \tr(\ca) $, where $\cs=\cp-\cn$, and $\ca=\cp+\cn$.
\end{enumerate}
\end{defi}
We  show in \secref{impleSDP} the existence of a one-sided oracle with  speed $S=\Omega(1)$ and error $\eps=0$, in which case the oracle only requires $\cs$ as input, instead of  $\cp$ and $\cn$. However, to construct such an oracle efficiently an additional error is introduced, which depends on $\cp+\cn$.


For the main algorithm $\mathtt{Sparsify}\left( \calM,\varepsilon \right)$, we  maintain
the matrix $A_j$ inductively as we 
 discussed at the beginning of \secref{overview}. By employing the $\oplus$ operator, we reduce the problem of constructing $\Delta_j$ with two-sided constraints to the problem of  constructing $\Delta_j\oplus\Delta_j$ with  one-sided constraints. 
We also use $\cs$  to  ensure that the length of $\ellip(\Delta)$ is large on the direction where the length of $\ellip(A_j)$ is small. See Algorithm~\ref{algo1} for formal description.

\begin{algorithm} \begin{algorithmic}[1]

\caption{$\mathtt{Sparsify}\left( \calM,\varepsilon\right)$ }

\State $j=0$, $\ma_0=\mzero$;

\State $\ell_{0}=-\frac{1}{4}$, $u_{0}=\frac{1}{4}$;

\While{$u_j - \ell_j < 1$}

\State $\mb_{j}=(u_{j}\mi-\ma_{j})^{2}\oplus(\ma_{j}-\ell_{j}\mi)^{2}$;

\State $\cp=  (1-2\varepsilon)(\ma_{j}-\ell_{j}\mi)^{-2}\exp\left(\ma_{j}-\ell_{j}\mi\right)^{-1}$; 

\State $\cn=  (1+2\varepsilon)(u_{j}\mi-\ma_{j})^{-2}\exp(u_{j}\mi-\ma_{j})^{-1}$;

\State $\Delta_{j}\oplus\Delta_{j}=\mathtt{Oracle}\left(\{\mm_{i}\oplus\mm_{i}\}_{i=1}^{m},\mb_{j},\frac{1}{2}\left(\cp\oplus\cp\right),\frac{1}{2}\left(\cn\oplus\cn\right)\right)$;

\State $\ensuremath{\ma_{j+1}\leftarrow\ma_{j}+\varepsilon\cdot\Delta_{j}}$;

\State $\delta_{u,j}=\varepsilon\cdot\frac{(1+2\varepsilon)(1+\eps)}{1-4\varepsilon}\cdot S\cdot\lambda_{\min}(\mb_{j})$
and $\delta_{\ell,j}=\varepsilon\cdot \frac{(1-2\varepsilon)(1-\eps)}{1+4\varepsilon}\cdot S\cdot\lambda_{\min}(\mb_{j})$;
\State $u_{j+1}\leftarrow u_{j}+\delta_{u,j}$, $\ell_{j+1}\leftarrow\ell_{j}+\delta_{\ell,j}$;

\State $j\leftarrow j+1.$

\EndWhile

\State \textbf{Return} $\ma_{j}$.
\label{algo1}
\end{algorithmic}
 \end{algorithm}

To analyse Algorithm~\ref{algo1}, we use the fact that the returned $\Delta_j$ satisfies the preconditions of \lemref{potentialchange} and prove in \lemref{potential_decreasing} that
\[
\E\left[\Phi_{u_{j+1},\ell_{j+1}}(\ma_{j+1})\right]\leq\Phi_{u_{j},\ell_{j}}(\ma_{j})
\]
for any iteration $j$.
Hence, with high probability the bounded ratio between $u_{T}$ and $\ell_T$ after the final iteration $T$ implies that the $\ellip(A_T)$ is close to be a sphere. In particular,  for any $\eps< 1/20$, a $(1+O(\eps))$-spectral sparsfier can be constructed by calling  $\oracle$  $O\left(\frac{\log^{2}n}{\varepsilon^{2} \cdot S}\right)$
times, which is  described in the lemma below.

\begin{lem}\label{lem:reductionlem}
Let $0<\varepsilon<1/20$. Suppose we have one-sided oracle $\mathtt{Oracle}$ with speed $S$ and error $\eps$. Then,
with constant probability 
 the algorithm
$\mathtt{Sparsify}\left( \calM,\varepsilon \right)$ outputs a $\left(1+O(\varepsilon)\right)$-spectral sparsifier with  $O\left(\frac{n}{\varepsilon^{2}\cdot S}\right)$ vectors
by calling $\oracle$ $O\left(\frac{\log^{2}n}{\varepsilon^{2} \cdot S}\right)$
times.
\end{lem}




\subsection{Solving $\oracle$ via SDP\label{sec:impleSDP}}


Now we  show that the required solution of  $\oracle(\calM, B, \cs)$ indeed exists\footnote{As the goal here is to prove the existence of $\oracle$ with error $\eps=0$, the input here is $\cs$ instead of $C_+$ and $C_-$.}, and can be further solved in nearly-linear time by  a semi-definite program.
We first prove that the required matrix satisfying the conditions of Definition~\ref{deforacle} exists for some absolute constant $S=\Omega(1)$ and $\eps=0$. To make a parallel discussion between Algorithm~\ref{algo1} and the algorithm we will present later,
 we  use $A$ to denote the output of the $\oracle$ instead of $\Delta$. We adopt the ideas between the ellipsoid and two spheres discussed before, but only  consider one sphere for the one-sided case. Hence, we introduce a barrier value $u_j$ for each iteration $j$, where $u_0=0$. We will use the  potential function 
\[
\Psi_{j}=\tr\left(u_jB-A_j \right)^{-1}
\]
in our analysis, where $u_j$  is increased by
$
\delta_j\triangleq (\Psi_j\cdot\lambda_{\min}(B))^{-1}
$
after  iteration $j$. 
Moreover, since we only need to prove the existence of the required matrix $A=\sum_{i=1}^m \alpha_i M_i$, our process proceeds for $T$ iterations, where \emph{only} one vector is chosen  in each iteration.  To find  a desired vector, we perform random sampling, where  each matrix $M_i$ is sampled with probability  $\mathsf{prob}(M_i)$ proportional to $M_i\bullet \cs$, i.e.,
\begin{equation}\label{eq:defdist}
\mathsf{prob}(M_i)\triangleq \frac{\left(M_i\bullet \cs\right)^+}{ \sum_{t=1}^m \left(M_t\bullet \cs\right)^+},
\end{equation}
where $x^+\triangleq \max\{x,0\}$. Notice that, since 
our goal is to construct $A$ such that $\mathbb{E}[\cs\bullet A]$ is lower bounded by some threshold as stated in Definition~\ref{deforacle}, we should not pick any matrix $M_i$ with $M_i\bullet \cs < 0$. This random sampling procedure is described in Algorithm~\ref{existencealgo}, and the properties of the output matrix is summarised in \lemref{one_sided}.

\begin{algorithm} \begin{algorithmic}[1]

\caption{$\mathtt{SolutionExistence}\left(\calM,\mb,\cs\right)$ \label{existencealgo} }

\State$\ma_{0}=\mzero$, $u_{0}=1$ and $T=\left\lfloor \lambda_{\min}(\mb)\cdot\tr(\mb^{-1})\right\rfloor $;


\For{$j = 0,1,\dots,T-1$}

\Repeat

\State Sample  a matrix $M_{t}$ with probability $\mathsf{prob}(M_{t})$;

\State Let $\Delta_{j}=(4\Psi_{j}\cdot \mathsf{prob}(M_{t}))^{-1}\cdot\mm_{t}$;

\Until{ $\Delta_{j}\preceq\frac{1}{2}(u_{j}\mb-\ma_{j})$;}

\State $\ma_{j+1}=\ma_{j}+\Delta_{j}$;

\State  $\delta_j= (\Psi_j\cdot\lambda_{\min}(B))^{-1}$;

\State $u_{j+1}=u_{j}+\delta_{j}$;

\EndFor
\State \textbf{Return} $\frac{1}{u_{T}}\ma_{T}$.

\end{algorithmic}
 \end{algorithm}

\begin{lem} \label{lem:one_sided}
Let $\mathbf{0}\preceq B\preceq I$ and $\cs$ be symmetric matrices, and $\calM=\{M_i\}_{i=1}^m$ be a set of $\textsf{PSD}$ matrices such that 
$\sum_{i=1}^m M_i = I$. Then 
$\mathtt{SolutionExistence}\left(\calM,\mb,\cs \right)$ outputs
a matrix $\ma=\sum_{i=1}^m\alpha_i \mm_i$ such that the following  holds:
\begin{enumerate}
\item $\mathrm{nnz}(\alpha)=\left\lfloor \lambda_{\min}(\mb)\cdot\tr(\mb^{-1})\right\rfloor $.
\item $\ma\preceq\mb$, and $\alpha_i \geq 0$ for all $i$.
\item $\E\left[\cs\bullet\ma\right]\geq\frac{1}{32}\cdot \lambda_{\min}(\mb)\cdot \tr\left( \cs\right)$.
\end{enumerate}
\end{lem}

\lemref{one_sided} shows that the required matrix $A$ defined in Definition~\ref{deforacle} exists, and can be found by  random sampling  described in Algorithm~\ref{existencealgo}. Our key observation is that such matrix $A$ can be constructed  by  a semi-definite program.

\begin{thm}\label{thm:sdp}
Let $\mathbf{0}\preceq B\preceq I$, $\cs$ be symmetric matrices, and $\calM=\{M_i\}_{i=1}^m$ be a set of matrices such that $\sum_{i=1}^m M_i = I$.
Let $S\subseteq [m]$ be a random set of $\left\lfloor \lambda_{\min}(\mb)\cdot\tr(\mb^{-1})\right\rfloor $ coordinates, where every index $i$ is picked with probability $\mathsf{prob}(M_i)$. Let $\ma^{\star}$
be the solution of the following semidefinite program
\begin{equation}\label{eq:SDP}
\max_{\alpha_i\geq0}\cs\bullet\left(\sum_{i\in S}\alpha_{i}\mm_{i}\right)\text{ subject to }\ma=\sum_{i\in S}\alpha_{i}\mm_{i}\preceq\mb
\end{equation}
Then, we have $\E \left[\cs\bullet\ma^{\star}\right] \geq \frac{1}{32}\cdot\lambda_{\min}(\mb)\cdot \tr (\cs)$.
\end{thm}

Taking the SDP formuation \eq{SDP} and the specific constraints of the $\oracle$'s input into account, the next lemma  shows that the required matrix used in each iteration of $\ensuremath{\mathtt{Sparsify}\left(\calM,\varepsilon\right)}$  can be computed efficiently by solving a semidefinite program.

\begin{lem}
\label{lem:oracle_implementation}The Oracle used in Algorithm $\ensuremath{\mathtt{Sparsify}\left(\calM,\varepsilon\right)}$
can be implemented in \[
\tilde{O}\left((Z+n^{\omega})\cdot\varepsilon^{-O(1)}\right)\quad\text{work}\quad\text{and}\quad\tilde{O}\left(\varepsilon^{-O(1)}\right)\text{ depth}
\]
where $Z=\sum_{i=1}^m\mathrm{nnz}(M_i)$ is the total number of non-zeros in $\mm_{i}$. When the
matrix $\sum_{i=1}^m M_i=\mi$ comes from spectral sparsification of graphs,
each iteration of $\ensuremath{\mathtt{Sparsify}\left(\calM,\varepsilon\right)}$
can be implemented in 
\[
\tilde{O}\left(m\varepsilon^{-O(1)}\right)\text{work}\quad\text{and}\quad\tilde{O}\left(\varepsilon^{-O(1)}\right)\text{ depth}.
\]
Furthermore, the speed of this one-sided oracle is $\Omega(1)$ and the error of this one-sided oracle is $\eps$. \end{lem}


Combining \lemref{reductionlem} and \lemref{oracle_implementation} gives us the proof of the main result.

\begin{proof}[Proof of \thmref{maingraph} and \thmref{mainmatrix}]

\lemref{oracle_implementation} shows that we can construct $\oracle$ with $\Omega(1)$ speed and $\eps$ error that runs in 
\[
\tilde{O}\left((Z+n^{\omega})\cdot\varepsilon^{-O(1)}\right)\quad\text{work}\quad\text{and}\quad\tilde{O}\left(\varepsilon^{-O(1)}\right)\text{ depth}
\]
for the matrix setting and 
\[
\tilde{O}\left(m\varepsilon^{-O(1)}\right)\text{work}\quad\text{and}\quad\tilde{O}\left(\varepsilon^{-O(1)}\right)\text{ depth}.
\]
for the graph setting. 
Combining this with   \lemref{reductionlem}, which states that  it suffices to call  $\oracle$ $\tilde{O}\left(1/\eps^2\right)$ times,
the main statements hold.
\end{proof}

\subsection{Further discussion}  

Before presenting a more detailed analysis of our algorithm,  we   compare our new approach with the previous ones for constructing a linear-sized spectral sparsifier, and see how we 
address the bottlenecks faced in previous constructions. Notice that all 
 previous algorithms require super poly-logarithmic number of  iterations, and super linear-time for each iteration. For instance, our previous algorithm~\cite{LS15} for constructing a sparsifier with $O(pn)$ edges requires $\Omega(n^{1/p})$ iterations and  $\Omega(n^{1+1/p})$ time per iteration for the following reasons:
\begin{itemize}
\item $\Omega\left(n^{1+1/p}\right)$ time is needed per iteration: Each iteration takes $n/g^{\Omega(1)}$ time to pick the vector(s) when $(\ell + g) I \preceq A \preceq (u-g) I$. To avoid eigenvalues of $A$ getting too close to the boundary  $u$ or $\ell$, i.e., $g$ being too small, we choose the potential function whose value dramatically increases when the eigenvalues of $A$ get close $u$ or $\ell$. As the cost, we need to scale down the added vectors by a $n^{1/p}$ factor. 
 \item $\Omega\left(n^{1/p}\right)$ iterations are needed:  By random sampling, we choose $O\left(n^{1-1/p}\right)$ vectors each iteration and use the matrix Chernoff bound to show that the ``quality'' of added  $O\left(n^{1-1/p}\right)$ vectors is just $p=\Theta(1)$ times worse than adding a single vector. Hence, this requires $\Omega\left(n^{1/p}\right)$ iterations.
\end{itemize}
In contrast, our new approach breaks these two barriers through the following way:
\begin{itemize}
\item A ``non-linear'' step: Instead of rescaling down the vectors we add uniformly, we pick much fewer vectors on the direction that blows up, i.e., we impose the condition $\Delta \preceq (uI-A)^2$ instead of $\Delta \preceq 1/p\cdot (uI-A)$. This allows us to use the new potential function \eq{defpotential} with form $\exp\left({x^{-1}}\right)$ to control the eigenvalues in a more aggressive way.
\item SDP filtering:  By matrix Chernoof bound, we know that the probability that we sample a few ``bad'' vectors is small. Informally,
 we apply semi-definite programming to filter out those bad vectors, and this allows us to add $\Omega\left(n/\log^{O(1)}(n)\right)$ vectors in each iteration.
\end{itemize}

\section{Detailed analysis \label{sec:details}}

In this section we give  detailed analysis for the statements presented in \secref{overview}.

\subsection{Analysis of the potential function}

Now we analyse the properties of the potential function \eq{defpotential}, and prove \lemref{potentialchange}.
 The following two facts from matrix analysis will be used in our analysis.

\begin{lem}[Woodbury Matrix Identity] \label{woodburg}
Let $A\in\mathbb{R}^{n\times n}$, $U\in\mathbb{R}^{n\times k}$, $C\in\mathbb{R}^{k\times k}$ and
$V\in\mathbb{R}^{k\times n}$ be matrices. Suppose that $A$, $C$ and $C^{-1} + VA^{-1}U$ are invertible, it holds that
\[
(A+UCV)^{-1}= A^{-1} - A^{-1}U\left( C^{-1} + VA^{-1}U \right)^{-1} VA^{-1}.
\]
\end{lem}

\begin{lem}[Golden-Thompson Inequality]\label{lem:Golden}
It holds for any symmetric matrices $A$ and $B$ that
\[
\tr\left(\mathrm{e}^{A+B}\right) \leq\tr\left(\mathrm{e}^{A}\cdot \mathrm{e}^B \right).
\]
\end{lem}

\begin{proof}[Proof of \lemref{potentialchange}]
We analyse the change of $\Phi_{u}(\cdot)$ and $\Phi_{\ell}(\cdot)$ individually.  First of all, notice that
\[
(uI-A-\Delta)^{-1}=(uI-A)^{-1/2}\left( I-(uI-A)^{-1/2} \Delta(uI-A)^{-1/2} \right)^{-1}(uI-A)^{-1/2}.
\]
We define $\Pi=(uI-A)^{-1/2}\Delta(uI-A)^{-1/2}$. Since $0\preceq\Delta\preceq \delta(uI-A)^2$ and $u-\ell \leq 1$, it holds that
\[
\Pi \preceq \delta (uI-A)\preceq \delta (uI-\ell I) \preceq \delta I,
\]
and therefore
\[
(I-\Pi)^{-1} \preceq I+ \frac{1}{1-\delta}\cdot\Pi.
\]
Hence, it holds that
\begin{align*}
(uI-A-\Delta)^{-1} &  \preceq (uI-A)^{-1/2}\left( I+ \frac{1}{1-\delta}\cdot\Pi \right)(uI-A)^{-1/2} \\
& = (uI-A)^{-1} + \frac{1}{1-\delta} \cdot (uI-A)^{-1}\Delta(uI-A)^{-1}.
\end{align*}
By the fact that $\tr\exp$ is monotone and the Golden-Thompson Inequality~(\lemref{Golden}), we have that
\begin{align*}
\Phi_u(A+\Delta) &= \tr\exp\left( (uI-A-\Delta)^{-1} \right) \\
& \leq \tr\exp\left( (uI-A)^{-1} + \frac{1}{1-\delta} \cdot (uI-A)^{-1}\Delta(uI-A)^{-1} \right) \\
& \leq \tr\left(\exp(uI- A)^{-1}\exp\left(\frac{1}{1-\delta}\cdot (uI- A)^{-1}\Delta(uI-A)^{-1}\right)\right).
\end{align*}
Since $0 \preceq \Delta\preceq \delta (uI-A)^{2}$ and $\delta \leq 1/10$ by assumption, we have that
$
(uI-A)^{-1}\Delta(uI-A)^{-1} \preceq \delta I$,
and
\[
\exp\left(\frac{1}{1-\delta}\cdot (uI-A)^{-1}\Delta(uI-A)^{-1}\right)\preceq I+(1+2\delta)\cdot (uI-A)^{-1}\Delta(uI-A)^{-1}.
\]
Hence, it holds that
\begin{align*}
\Phi_{u}(\ma+\Delta) & \leq  \tr\left(\ce^{(uI- A)^{-1}}\cdot \left( I+(1+2\delta)(uI-A)^{-1}\Delta(uI-A)^{-1}\right) \right)\\
 & =  \Phi_{u}(\ma)+(1+2\delta)\cdot\tr(\mathrm{e}^{(uI-A )^{-1}}(uI-A)^{-2}\Delta).
\end{align*}

By the same analysis, we have that 
\begin{align*}
\Phi_{\ell}(\ma+\Delta) & \leq  \tr\left(\ce^{(A-\ell I)^{-1}}\cdot \left( I-(1-2\delta)(A- \ell I)^{-1}\Delta(A-\ell I)^{-1}\right) \right)\\
 & =  \Phi_{\ell}(\ma)-(1-2\delta)\cdot\tr(\mathrm{e}^{(A-\ell I )^{-1}}(A- \ell I)^{-2}\Delta).
\end{align*}
Combining the analysis on $\Phi_u(A+\Delta)$ and $\Phi_{\ell}(A+\Delta)$ finishes the proof. 
\end{proof}

\begin{lem}\label{lem:potential2}
Let $A$ be a symmetric matrix. Let $u,\ell$ be the barrier values such that $u-\ell\leq1$ and $\ell\mi\prec\ma\prec u\mi$. Assume
that $0\leq \delta_u\leq \delta\cdot \lambda_{\min}(uI-A)^2$ and $0\leq \delta_{\ell}\leq\delta\cdot \lambda_{\min}(A-\ell I)^2$ for $\delta\leq 1/10$. Then, it holds that 
\begin{align*}
\Phi_{u+\delta_u,\ell+\delta_{\ell}}(\ma)  \leq  & \Phi_{u,\ell}(\ma)  - (1-2\delta) \delta_{u}\cdot \tr\left(\ce^{(u\mi-\ma)^{-1}}(u\mi-\ma)^{-2}\right)\\
 &  \qquad +(1+2\delta)\delta_{\ell}\cdot \tr\left(\ce^{(\ma-\ell\mi)^{-1}}(\ma-\ell\mi)^{-2}\right).
\end{align*}
\end{lem}

\begin{proof}
Since $0\leq \delta_u\leq \delta\cdot \lambda_{\min}(uI-A)^2$ and $0\leq \delta_{\ell}\leq\delta\cdot \lambda_{\min}(A-\ell I)^2$,
 we have that $\delta_u\cdot I\preceq\delta\cdot (uI-A)^2$ and
$\delta_{\ell}\cdot I\preceq\delta\cdot (A-\ell I)^2$.
The statement follows by a similar analysis for proving \lemref{potentialchange}.
\end{proof}

\subsection{Analysis of the reduction}

Now we present the detailed analysis for the reduction from a spectral sparsifier to a one-sided oracle. We first analyse Algorithm~\ref{algo1}, and prove that in expectation the value of the potential function is not increasing. Based on this fact, we will give a proof of \lemref{reductionlem}, which shows that a $(1+O(\varepsilon))$-spectral sparsifier can be constructed by calling $\oracle$ $O\left(\frac{\log^{2}n}{\varepsilon^{2} \cdot S}\right)$
times.

\begin{lem}
\label{lem:potential_decreasing} Let $A_j$ and $A_{j+1}$ be the matrices constructed by Algorithm~\ref{algo1} in iteration $j$ and $j+1$, and 
assume that $0 \leq \varepsilon\leq 1/20$. Then, it holds that
\[
\E\left[\Phi_{u_{j+1},\ell_{j+1}}(\ma_{j+1})\right]\leq\Phi_{u_{j},\ell_{j}}(\ma_{j}).
\]
\end{lem}
\begin{proof}
By the description of Algorithm~\ref{algo1} and Definition~\ref{deforacle}, it holds that
\[
\Delta_j\oplus\Delta_j \preceq (u_{j}\mi-\ma_{j})^{2}\oplus(\ma_{j}-\ell_{j}\mi)^{2},
\]
which implies that
$\Delta_{j}\preceq  (u_{j}\mi-\ma_{j})^{2}$ and $\Delta_{j}\preceq  (\ma_{j}-\ell_{j}\mi)^{2}$. Since $u_j-\ell_j \leq 1$ by the algorithm description and $0\leq\eps\leq 1/20$, by   setting $\Delta=\varepsilon\cdot\Delta_j$ in \lemref{potentialchange}, we have 
\begin{align*}
 \Phi_{u_j, \ell_j} (\ma_{j} + \varepsilon\cdot \Delta_j)
& \leq  \Phi_{u_{j},\ell_{j}}(\ma_{j})   + \varepsilon (1+2\varepsilon)\cdot \tr\left(\mathrm{e}^{(u_{j}\mi-\ma_{j})^{-1}}(u_{j}\mi-\ma_{j})^{-2}\Delta_{j}\right)\\
 &  \qquad\qquad - \varepsilon (1-2\varepsilon)\cdot \tr\left(\mathrm{e}^{(\ma_{j}-\ell_{j}\mi)^{-1}}(\ma_{j}-\ell_{j}\mi)^{-2}\Delta_{j}\right)\\
 & = \Phi_{u_{j},\ell_{j}}(\ma_{j})-\varepsilon\cdot\cs\bullet\Delta_{j}.
\end{align*}
Notice that the matrices $\{M_i\oplus M_i\}_{i=1}^m$ as the input of $\mathtt{Oracle}$ always satisfy $\sum_{i=1}^mM_i\oplus M_i=I\oplus I$.
Using this and the definition of $\mathtt{Oracle}$, we know that
\[
\E\left[\cs\bullet\Delta_j\right] \geq 
 S\cdot\lambda_{\min}(B_{j})\cdot\tr (\cs)-\eps\cdot S\cdot\lambda_{\min}(B_{j})\cdot\tr\left( \ca\right),
\]
Let $\alpha_{j}=\varepsilon\cdot S\cdot\lambda_{\min}(\mb_{j})$. Then we have that 
\begin{align}
\E\left[\Phi_{u_{j},\ell_{j}}(\ma_{j+1}) \right] &\leq \Phi_{u_{j},\ell_{j}}(\ma_{j}) - \varepsilon\cdot \E\left[\cs\bullet\Delta_{j}\right]\nonumber  \\
& \leq \Phi_{u_{j},\ell_{j}}(\ma_{j})  +  (1+2\varepsilon) (1+\eps)\cdot \alpha_j\cdot \tr\left(\mathrm{e}^{(u_{j}\mi-\ma_{j})^{-1}}(u_{j}\mi-\ma_{j})^{-2}\right) \nonumber\\
 &  \qquad\qquad - (1-2\varepsilon)  (1-\eps)\cdot \alpha_j\cdot \tr\left(\mathrm{e}^{(\ma_{j}-\ell_{j}\mi)^{-1}}(\ma_{j}-\ell_{j}\mi)^{-2}\right).\label{eq:echange1}
\end{align}

On the other hand, using that $0\leq\eps \leq 1/20$, $S \leq 1$ and $\Delta_{j}\preceq  (u_{j}\mi-\ma_{j})$, we have that
\[
\delta_{u,j} \leq \varepsilon\cdot\frac{(1+2\varepsilon)(1+\eps)}{1-4\varepsilon} \cdot \lambda_{\min}(u_{j}\mi-\ma_{j})^{2}\leq 2 \eps \cdot \lambda_{\min}(u_{j}\mi-\ma_{j+1})^{2}
\]
and
\[
\delta_{\ell,j} \leq \varepsilon\cdot\frac{(1-2\varepsilon)(1-\eps)}{1+4\varepsilon} \cdot \lambda_{\min}(\ma_{j}-\ell_{j}\mi)^{2}\leq 2 \eps \cdot \lambda_{\min}(\ma_{j+1}-\ell_{j}\mi)^{2}.
\]
Hence, \lemref{potential2} shows that
\begin{align}
\Phi_{u_{j} +\delta_{u,j},\ell_{j}+\delta_{\ell,j}}(\ma_{j+1})& \leq  \Phi_{u_{j},\ell_{j}}(\ma_{j+1})
 - (1 - 4\varepsilon)\delta_{u,j}\cdot \tr\left(\mathrm{e}^{(u_{j}\mi-\ma_{j+1})^{-1}}(u_{j}\mi-\ma_{j+1})^{-2}\right)\nonumber \\
 & \qquad \qquad\qquad  + (1 + 4\varepsilon)\delta_{\ell,j}
\cdot \tr\left(\mathrm{e}^{(\ma_{j+1}-\ell_{j}\mi)^{-1}}(\ma_{j+1}-\ell_{j}\mi)^{-2}\right) \nonumber \\
 & \leq  \Phi_{u_{j},\ell_{j}}(\ma_{j+1}) - (1 - 4\varepsilon)\delta_{u,j}\cdot \tr\left(\mathrm{e}^{(u_{j}\mi-\ma_{j})^{-1}}(u_{j}\mi-\ma_{j})^{-2}\right)\nonumber\\
 &  \qquad\qquad\qquad + (1 + 4\varepsilon)\delta_{\ell,j}\cdot \tr\left(\mathrm{e}^{(\ma_{j}-\ell_{j}\mi)^{-1}}(\ma_{j}-\ell_{j}\mi)^{-2}\right).\label{eq:echange2}
\end{align}
By combining \eq{echange1}, \eq{echange2}, and setting $(1-4\varepsilon)\delta_{u,j}=(1+2\varepsilon)(1+\eps)\alpha_{j}$,
 $(1+4\varepsilon)\delta_{\ell,j}=(1-2\varepsilon)(1-\eps)\alpha_{j}$, we have that
$
\mathbb{E}\left[\Phi_{u_{j+1},\ell_{j+1}}(\ma_{j+1})\right]\leq\Phi_{u_{j},\ell_{j}}(\ma_{j}).
$
\end{proof}

\begin{proof}[Proof of \lemref{reductionlem}]
We first bound the number of times the algorithm calls the oracle. 
Notice that
$
\Phi_{u_0, \ell_0} = 2\cdot \tr\exp\left(\left(\frac{1}{4}I\right)^{-1}\right) = 2 \ce^4\cdot n.
$
Hence, by \lemref{potential_decreasing} we have  $\mathbb{E}\left[ \Phi_{u_j, \ell_j} (A_j) \right]=O(n)$ for any iteration $j$. By Markov's inequality, it holds that $ \Phi_{u_j, \ell_j}   (A_j)=n^{O(1)}$ with high probability in $n$.
In the remainder of the proof, we assume that this event occurs.

Since $\mb_{j}=(u_{j}\mi-\ma_{j})^{2}\oplus(\ma_{j}-\ell_{j}\mi)^{2}$ by definition, it holds that 
\[
\exp\left(  (\lambda_{\min}\left(\mb_j \right))^{-1/2}\right) \leq \Phi_{u_j, \ell_j}(A_j) =   n^{O(1)},
\]
which implies that 
\begin{equation}\label{eq:lblambdab}
\lambda_{\min}\left( \mb_j \right)=\Omega\left(\log^{-2}n\right).
\end{equation}
On the other hand, in iteration $j$ the gap between $u_j$ and $\ell_j$ is increased by
\begin{equation}\label{eq:gapincrease}
\delta_{u,j}-\delta_{\ell,j}=\Omega\left(\varepsilon^{2}\cdot S\cdot\lambda_{\min}(\mb_{j}) \right).
\end{equation}
Combining this with \eq{lblambdab} gives us that
 \begin{equation*}
\delta_{u,j}-\delta_{\ell,j}  =  \Omega\left(\frac{\varepsilon^{2}\cdot S}{\log^{2}n}\right)
\end{equation*}
for any $j$.
Since $u_0-\ell_0=1/2$ and the algorithm terminates once $u_j - \ell_j>1$ for some $j$, with high probability in $n$, the algorithm terminates in
$
O\left( \frac{\log^2 n}{\varepsilon^2\cdot S}  \right)
$
iterations.

Next we prove that the number of $M_i$'s involved in the output is at most  $O\left(\frac{n}{\eps^2\cdot S} \right)$.
By the properties of $\mathtt{Oracle}$, the number of matrices in iteration $j$ is at most 
$\lambda_{\min}(B_j)\cdot\tr( B_j^{-1})$. Since $x^{-2}\leq \exp\left( x^{-1}\right)$ for all $x>0$, it holds for any iteration $j$ that
\begin{align*}
\tr\left(\mb_{j}^{-1}\right) & =\tr \left(\left(u_{j}\mi-\ma_{j}\right)^{-2}\right) +\tr \left(\left(\ma_{j}-\ell_{j}\mi\right)^{-2}\right) \leq\Phi_{u_{j},\ell_{j}}(\ma_{j})
\end{align*}
By \eq{gapincrease}, we know that for added matrix $M_i$ in iteration $j$, the average increase of the  gap $u_j-\ell_j$ for each added matrix is 
$
\Omega\left(\frac{\varepsilon^{2}\cdot S}{\Phi_{u_{j},\ell_{j}}(\ma_j)} \right)$.
Since $\mathbb{E}\left[\Phi_{u_j, \ell_j}(A_j)\right]=O(n)
$, for every new added matrix, in expectation the gap between $u_j$ and $\ell_j$ is increased by
$\Omega\left(\frac{\varepsilon^{2}\cdot S}{ n} \right)$.
By the ending condition of the algorithm, i.e., $u_j-\ell_j>1$, and Markov's inequality, the number of matrices picked in total is at most 
$
O \left(\frac{n}{\varepsilon^2\cdot S} \right)
$
with constant probability.

Finally we prove that the output is a $(1+O(\varepsilon))$-spectral sparsifier. Since the condition number of the output matrix $A_j$ is at most
\[
\frac{u_j}{\ell_j} =
\left( 1 - \frac{u_j-\ell_j}{u_j} \right)^{-1},
\]
it suffices to prove that
$(u_j-\ell_j)/u_j=O(\varepsilon)$ and this  easily follows from the ending condition of the algorithm and
\[
\frac{\delta_{u,j} -\delta_{\ell, j}}{\delta_{u,j}}= O(\eps).
\]
\end{proof}


\subsection{Existence proof for $\oracle$}

\begin{proof}[Proof of \lemref{one_sided}]
The property on $\mathrm{nnz}\left(\alpha\right)$ follows from the algorithm description.
For the second property, notice that  every chosen matrix $\Delta_j$ in iteration $j$ satisfies 
$\Delta_{j}\preceq\frac{1}{2}(u_{j}\mb-\ma_{j})$,
which implies that 
$\ma_{j}\preceq u_{j}\mb$ holds for any iteration $j$. Hence,
$\ma=\frac{1}{u_{T}}\ma_{T}\preceq\mb$, and $\alpha_i \geq 0$ since $\Psi_{j} \geq 0$.

Now we prove the third statement. Let 
\[
\beta=\sum_{i=1}^m \left(\mm_{i}\bullet\cs\right)^+.
\]
Then, for  each matrix $\mm_{i_{j}}$ picked in iteration $j$, $\mc\bullet\ma_{j}$ is increased by
\begin{align*}
\cs\bullet\Delta_{j} & =\frac{1}{4\Psi_{j}\cdot \mathsf{prob}(M_{i_j})}\cdot \cs\bullet\mm_{i_{j}}=\frac{\beta}{4\Psi_{j}}.
\end{align*}
On the other hand, it holds that
\[
u_T= u_0 + \sum_{j=0}^{T-1} \delta_j = 1+\sum_{j=0}^{T-1} \left( \Psi_j\cdot\lambda_{\min}(B) \right)^{-1}.
\]
Hence,
we have that
\begin{align}
\cs\bullet\ma & =\frac{1}{u_T}\cdot C\bullet \left(\sum_{j=0}^{T-1} \Delta_j \right) 
= \frac{\sum_{j=0}^{T-1}\beta\cdot(4\Psi_{j})^{-1}}{1+\sum_{j=0}^{T-1} (\Psi_{j}\cdot\lambda_{\min}(\mb))^{-1}} \nonumber\\
& =  \frac{\beta\lambda_{\min}(\mb) }{4}\cdot\frac{\sum_{j=0}^{T-1}\Psi_{j}^{-1}}{\lambda_{\min}(\mb) +\sum_{j=0}^{T-1}\Psi_{j}^{-1}}\nonumber \\
 & \geq \frac{\beta\lambda_{\min}(\mb)} {4}\cdot\frac{\sum_{j=0}^{T-1}(\Psi_{j}+\Psi_{0})^{-1}}{\lambda_{\min}(\mb) +\sum_{j=0}^{T-1}(\Psi_{j}+\Psi_{0})^{-1}} \nonumber \\
 & \geq \frac{\beta\lambda_{\min}(\mb) }{4}\cdot\frac{\sum_{j=0}^{T-1}(\Psi_{j}+\Psi_{0})^{-1}}{\lambda_{\min}(\mb) + T \cdot \Psi_{0}^{-1}} 
 \geq \frac{\beta}{8}\cdot\sum_{j=0}^{T-1}\left(\Psi_{j}+\Psi_{0}\right)^{-1}\label{eq:CA_bound},
\end{align}
where the last inequality follows by the choice of $T$.
Hence, it suffices to bound $\Psi_j$. 

Since $\Delta_{j}\preceq\frac{1}{2}(u_{j}\mb-\ma_{j})\preceq\frac{1}{2}(u_{j+1}\mb-\ma_{j})$,
we have that 
\begin{equation}\label{eq:mideq1}
\Psi_{j+1} \leq \tr \left((u_{j+1}\mb-\ma_{j})^{-1} \right)+2\cdot \Delta_{j}\bullet(u_{j+1}\mb-\ma_{j})^{-2}.
\end{equation}
Since $\tr(u\mb-\ma_{j})^{-1}$ is convex in $u$, we have that
\begin{equation}\label{eq:mideq2}
\tr\left((u_{j} \mb-\ma_{j})^{-1}\right) \geq \tr\left((u_{j+1} \mb-\ma_{j})^{-1}\right) + \delta_j \tr\left((u_{j+1} \mb-\ma_{j})^{-2} B\right)
\end{equation}
Combining \eq{mideq1} and \eq{mideq2}, we have that 
\begin{align}
\Psi_{j+1}  & \leq \tr \left(\left(u_jB-A_j\right)^{-1}\right) -\delta_j\cdot\mathrm{tr}\left( \left( u_{j+1} B -A_j\right)^{-2}B \right) +2\cdot \Delta_{j}\bullet(u_{j+1}\mb-\ma_{j})^{-2} \nonumber\\
& = \Psi_j  -\delta_j\cdot\lambda_{\min}(B)\cdot \tr\left((u_{j+1}B-A_j)^{-2}\right)+2\cdot \Delta_{j}\bullet(u_{j+1}\mb-\ma_{j})^{-2}  \label{eq:upphinext}
\end{align}
Let $\mathcal{E}_{j}$ be the event that $\Delta_{j}\preceq\frac{1}{2}(u_{j}\mb-\ma_{j})$.  Notice that our picked $\Delta_j$ in each iteration always satisfies $\mathcal{E}_{j}$ by algorithm description.
 Since 
\[
\E\left[\Delta_{j}\bullet(u_{j}\mb-\ma_{j})^{-1}\right]=\sum_{i: (M_i\bullet C)^+>0}
\mathsf{prob}(M_i)\cdot
\frac{1}{4\Psi_{j}\cdot\mathsf{prob}(M_i)  }\cdot\mm_{i}\bullet(u_{j}\mb-\ma_{j})^{-1}\leq\frac{1}{4},
\]
by Markov inequality it holds that 
\[
\mathbb{P}\left[\mathcal{E}_{j} \right]= \mathbb{P}\left[ \Delta_{j}\preceq\frac{1}{2}(u_{j}\mb-\ma_{j})\right] \geq\frac{1}{2},
\]
and therefore
\begin{align*}
\E\left[\Delta_{j}\bullet(u_{j+1}\mb-\ma_{j})^{-2}~|~\mathcal{E}_{j}\right] & \leq\frac{\E\left[\Delta_{j}\bullet(u_{j+1}\mb-\ma_{j})^{-2}\right]}{\mathbb{P}\left(\mathcal{E}_{j}\right)} \\
& \leq 2\cdot \E\left[\Delta_{j}\bullet(u_{j+1}\mb-\ma_{j})^{-2}\right]\\
 &= \frac{1}{2\Psi_{j}}\cdot \sum_{i: (M_i\bullet C)^+>0}\mm_{i}\bullet(u_{j+1}\mb-\ma_{j})^{-2}\\
 &\leq \frac{1}{2\Psi_{j}} \cdot \tr(u_{j+1}\mb-\ma_{j})^{-2}.
\end{align*} 
Combining the inequality above,  \eq{upphinext}, and the fact that every $\Delta_j$ picked by the algorithm satisfies $\mathcal{E}$, we have that 
$$
\E\left[ \Psi_{j+1} \right]\leq \Psi_{j}+\left(\frac{1}{\Psi_{j}}-\delta_{j}\cdot \lambda_{\min}(\mb)\right)\cdot \tr\left(u_{j+1}\mb-\ma_{j} \right)^{-2}.
$$
By our choice of $\delta_j$, it holds for any iteration $j$ that $\E\left[ \Psi_{j+1}\right] \leq \Psi_j$, and 
\[
\E\left[ \left(\Psi_{j+1} + \Psi_0 \right)^{-1}\right]\geq \E\left(\Psi_{j} + \Psi_0 \right)^{-1} \geq \frac{1}{2\cdot \Psi_0}. 
\]
Combining this with \eq{CA_bound}, it holds that 
\[\E\left[\cs\bullet\ma\right] \geq \frac{\beta}{8}\sum_{j=0}^{T-1}\E\left[(\Psi_{j}+\Psi_{0})^{-1} \right]
 \geq \frac{ \beta }{16}\cdot\frac{T}{\Psi_0} = \frac{ \beta }{16}\cdot\frac{T}{\tr\left( \mb^{-1}\right)} \geq \frac{ \tr(\cs) }{16}\cdot\frac{T}{\tr\left( \mb^{-1} \right)} .
\]
The result follows from the fact that $T \geq \lambda_{\min}(\mb) \tr\left( \mb^{-1}\right)/2$.
\end{proof}
Using the lemma above, we can prove that such $\ma$ can be solved
by a semidefinite program.

\begin{proof}[Proof of \thmref{sdp}]
Note that the probability we used in the statement is the same as the probability we used in $\mathtt{SolutionExistence}\left(
\mathcal{M},\mb,\cs\right)$. Therefore, Lemma \ref{lem:one_sided} shows that there is a matrix $A$ of the form $\sum_{i=1}^m\alpha_{i}\mm_{i}$ such that 
\[
\E\left[ \cs\bullet\ma\right] \geq \frac{1}{32}\cdot\lambda_{\min}(\mb)\cdot  \tr\left(\cs\right),
\]
The statement follows by the fact that $A^{\star}$ is the solution of the semidefinite program \eq{SDP} that maximises $\cs\bullet\ma$.
\end{proof}

\subsection{Implementing the SDP in nearly-linear time\label{sec:sdp}}

Now, we discuss how to solve the SDP \eq{SDP} in nearly-linear time.
Since this SDP is a packing SDP, it is known how to solve it in nearly-constant depth \cite{jain2011parallel,allen2016using,peng2012faster}.
The following result will be used in our analysis.
\begin{thm}
[\cite{allen2016using}]\label{thm:solve_SDP}Given a SDP 
\[
\max_{x\geq0}c^{\rot}x\text{ subject to }\sum_{i=1}^{m}x_{i}\ma_{i}\preceq\mb
\]
with $\ma_{i}\succeq\mzero$, $\mb\succeq\mzero$ and $c\in\R^{m}$.
Suppose that we are given a direct access to the vector $c\in\Rset^m$ and an
indirect access to $\ma_{i}$ and $\mb$ via an oracle $\mathcal{O}_{L,\delta}$
which inputs a vector $x\in\R^{m}$ and outputs a vector $v\in\R^{m}$
such that
\[
v_{i}\in\left(1\pm\frac{\delta}{2}\right)\left[\ma_{i}\bullet\mb^{-1/2}\exp\left(L\cdot\mb^{-1/2}\left(\sum_{i}x_{i}\ma_{i}-\mb\right)\mb^{-1/2}\right)\mb^{-1/2}\right]
\]
in $\mathcal{W}_{L,\delta}$ work and $\mathcal{D}_{L,\delta}$ depth for any $x$ such that $x_i\geq0$
and $\sum_{i=1}^{m}x_{i}\ma_{i}\preceq2\mb$. Then, we can output
$x$ such that 
\[
\E\left[c^{\rot}x \right]\geq(1-O(\delta))\mathsf{OPT}\quad\text{with}\quad\sum_{i=1}^{m}x_{i}\ma_{i}\preceq\mb
\]
 in
\[
O\left(\mathcal{W}_{L,\delta}\log m\cdot\log\left(nm/\delta\right)/\delta^{3}\right) \text{work}\quad\text{and}\quad O\left(\mathcal{D}_{L,\delta}\log m\cdot\log(nm/\delta)/\delta\right) \text{depth},
\]
where $L=(4/\delta)\cdot\log(nm/\delta)$.
\end{thm}
Since we are only interested in a fast implementation of the one-sided oracle used 
in $\ensuremath{\mathtt{Sparsify}\left(\calM,\varepsilon\right)}$,
it suffices to solve the SDP \eq{SDP} for this particular situation. 

\begin{proof}[Proof of \lemref{oracle_implementation}]
Our basic idea is to use \thmref{solve_SDP} as the one-sided oracle.
Notice that each iteration of $\ensuremath{\mathtt{Sparsify}\left(\calM,\varepsilon\right)}$
uses the one-sided oracle with the input 
\begin{align*}
\cp & =\frac{1-2\varepsilon}{2}\left((\ma-\ell\mi)^{-2}\exp(\ma-\ell\mi)^{-1}\oplus(\ma-\ell\mi)^{-2}\exp(\ma-\ell\mi)^{-1}\right),\\
\cn & =\frac{1+2\varepsilon}{2}\left((u\mi-\ma)^{-2}\exp(u\mi-\ma)^{-1}\oplus(u\mi-\ma)^{-2}\exp(u\mi-\ma)^{-1}\right),\\
B & =(u\mi-\ma)^{2}\oplus(\ma-\ell\mi)^{2},
\end{align*}
where we drop the subscript $j$ indicating iterations here for simplicity.
To apply~\thmref{sdp}, we first sample a subset $S\subseteq [m]$,
then solve the SDP
\[
\max_{
\beta_{i}\geq0,\beta_{i}=0\text{ on }i\notin S
}(\cp-\cn)\bullet\left(\sum_{i=1}^m\beta_{i}\mm_{i}\oplus\mm_{i}\right)\quad \emph{ subject to } \sum_{i=1}^m\beta_{i}\mm_{i}\oplus\mm_{i}\preceq\mb.
\]
By ignoring the matrices $\mm_{i}$ with $i\not\in S$, this SDP
is equivalent to the SDP 
\[
\max_{
\beta_{i}\geq0}c^{\top}\beta \quad\emph{ subject to }\sum_{i=1}^m\beta_{i}\mm_{i}\oplus\mm_{i}\preceq\mb
\]
where $c_{i}=(\cp-\cn)\bullet(\mm_{i}\oplus\mm_{i})$. Now assume that  
 (i) we can approximate $c_{i}$ with $\delta(\cp+\cn)\bullet(\mm_{i}\oplus\mm_i)$ additive
error, and (ii) for any $x$ such that $\sum_{i=1}^mx_{i}\mm_{i}\oplus\mm_{i}\preceq2\mb$ and $L =\tilde{O}(1/\delta)$,
we can approximate
\[
(\mm_{i}\oplus\mm_{i})\bullet\mb^{-1/2}\exp\left(L\cdot\mb^{-1/2}\left(\sum_{i}x_{i}\mm_{i}\oplus\mm_{i}-\mb\right)\mb^{-1/2}\right)\mb^{-1/2}
\]
with $1\pm\delta$ multiplicative error.
Then, by \thmref{solve_SDP} we can find a vector $\beta$
such that
\begin{align*}
\E\left[ c^{\rot}\beta\right] & \geq(1-O(\delta))\mathsf{OPT}-\delta\sum_{i=1}^m\beta_{i}(\cp+\cn)\bullet(\mm_{i}\oplus\mm_{i})\\
 & \geq(1-O(\delta))\mathsf{OPT}-\delta\sum_{i=1}^m\beta_{i}(\cp+\cn)\bullet\mb\\
 & \geq\mathsf{OPT}-O(\delta)(\cp+\cn)\bullet\mb.
\end{align*}
where we used that $\mm_{i}\oplus\mm_{i}\preceq\mb$ and $\mathsf{OPT}\leq(\cp+\cn)\bullet\mb$.
Since $u-\ell\leq1$, we have that $B\preceq I\oplus I$ and hence
\begin{align*}
\E\left[ c^{\rot}\beta\right] & \geq\mathsf{OPT}-O(\delta)\cdot\tr(\cp+\cn)\\
 & \geq\frac{1}{32}\lambda_{\min}(\mb)\cdot\left(\tr(\mc)-O(\delta\log^{2}n)\cdot\tr(\cp+\cn)\right)
\end{align*}
where we apply~\thmref{sdp} and (\ref{eq:lblambdab}) at
the last line. Therefore, this gives an oracle with speed $1/32$
 and $\varepsilon$ error by setting $\delta=\varepsilon/\log^{2}n$.

The problem of approximating sample probabilities, $\{c_{i}\}$, as well as implementing the oracle $\mathcal{O}_{L,\delta}$
is similar with approximating leverage scores~\cite{SS}, and relative leverage scores~\cite{zhu15,LS15}. All these references  use 
the Johnson-Lindenstrauss lemma
to reduce the problem of approximating
matrix dot product or trace to matrix vector multiplication. The only
difference is that, instead of computing $(\ma-\ell\mi)^{-(q+1)}x$ and $(u\mi-\ma)^{-(q+1)}x$ for a given vector $x$ in other references, we compute $(\ma-\ell\mi)^{-2}\exp(\ma-\ell\mi)^{-1}x$ and 
$(u\mi-\ma)^{-2}\exp(u\mi-\ma)^{-1}x$. These can be approximated by Taylor expansion and the number of terms required for Taylor expansion depends on how close the eigenvalues
of $A$ are to the boundary ($u$ or $\ell$). In particular, we show in \secref{taylor} that $\tilde{O}(1/g^2)$ terms in Taylor expansion suffices,
where the gap $g$ is the largest number such that $(\ell +g)I \preceq A \preceq (u - g)I$. Since $1/g^2 = \tilde{O}(1)$ by (\ref{eq:lblambdab}), each iteration can be implemented via solving $\tilde{O}\left(1/\eps^{O(1)}\right)$  linear systems and $\tilde{O}\left(1/\eps^{O(1)}\right)$  matrix vector multiplication. For the matrices coming from graph sparsification, this can be done in nearly-linear work and nearly-constant depth~\cite{PengS14,KyngLPSS16}. For general matrices, this can be done in input sparsity time and nearly-constant depth~\cite{nelson2013osnap,li2013iterative,cohen2015uniform}.
\end{proof}

\subsection{Taylor Expansion of $x^{-2} \exp(x^{-1})$ \label{sec:taylor}}

\begin{theorem}
[Cauchy's Estimates]\label{thm:cauchy_estimate}Suppose $f$ is holomorphic
on a neighborhood of the ball $B\triangleq\{z\in\mathbb{C}\ :\ \left|z-s\right|\leq r\}$,
then we have that
\[
\left|f^{(k)}(s)\right|\leq\frac{k!}{r^{k}}\sup_{z\in B}\left|f(z)\right|.
\]
\end{theorem}
\begin{lemma}
\label{lem:apr_exp}Let $f(x)=x^{-2}\exp(x^{-1})$. For any $0<x\leq 1$,
we have that
\[
\left|f(x)-\sum_{k=0}^{d}\frac{1}{k!}f^{(k)}(1)(x-1)^{k}\right|\leq8(d+1)\mathrm{e}^{\frac{5}{x}-xd}.
\]
In particular, if $d\geq\frac{c}{x^{2}}\log(\frac{1}{x\varepsilon})$
for some large enough universal constant $c$, we have that
\[
\left|f(x)-\sum_{k=0}^{d}\frac{1}{k!}f^{(k)}(1)(x-1)^{k}\right|\leq\varepsilon.
\]
\end{lemma}
\begin{proof}
By the formula of the remainder term in Taylor series, we have that
\[
f(x)=\sum_{k=0}^{d}\frac{1}{k!}f^{(k)}(1)(x-1)^{k}+\frac{1}{d!}\int_{1}^{x}f^{(d+1)}(s)(x-s)^{d}ds.
\]
For any $s\in[x,1]$, we define $D(s)=\{z\in\mathbb{C}\ :\ \left|z-s\right|\leq s-\frac{x}{2}\}$.
Since $\left|f(z)\right|\leq(x/2)^{-2}\exp(2/x)$ on $z\in D(s)$,
Cauchy's estimates (Theorem \ref{thm:cauchy_estimate}) shows that
\[
\left|f^{(d+1)}(s)\right|\leq\frac{(d+1)!}{(s-\frac{x}{2})^{d+1}}\sup_{z\in B(s)}\left|f(z)\right|\leq\frac{(d+1)!}{(s-\frac{x}{2})^{d+1}}\frac{4}{x^{2}}\exp\left(\frac{2}{x}\right).
\]
Hence, we have that
\begin{align*}
\left|f(x)-\sum_{k=0}^{d}\frac{1}{k!}f^{(k)}(t)(x-t)^{k}\right|\leq & \frac{1}{d!}\left|\int_{1}^{x}\frac{(d+1)!}{(s-\frac{x}{2})^{d+1}}\frac{4}{x^{2}}\exp\left(\frac{2}{x}\right)(x-s)^{d}\mathrm{d}s\right|\\
= & \frac{4(d+1)\mathrm{e}^{\frac{2}{x}}}{x^{2}}\int_{x}^{1}\frac{(s-x)^{d}}{(s-\frac{x}{2})^{d+1}}\mathrm{d}s\\
\leq & \frac{8(d+1)\mathrm{e}^{\frac{2}{x}}}{x^{3}}\int_{x}^{1}\left(1-x\right)^{d}\mathrm{d}s\\
\leq & 8(d+1)\cdot \mathrm{e}^{\frac{5}{x}-xd}.
\end{align*}
\end{proof}

\section*{Acknowledgement}
The authors  would like to thank Michael Cohen for helpful discussions and suggesting the ideas to improve the sparsity from $O(n/\eps^3)$ to $O(n/\eps^2)$.

\bibliography{reference}

\appendix

\end{document}